\pdfoutput=1
\documentclass[aps,pra,superscriptaddress,twocolumn,showpacs,twoside,a4paper]{revtex4-1}
\usepackage{graphicx}
\usepackage{amsmath}
\usepackage{amsthm}
\usepackage{amssymb}
\usepackage{amsfonts}
\usepackage{bbm}
\usepackage[utf8]{inputenc}
\usepackage{hyperref}
\usepackage{version}
\graphicspath{{Images/}}
\usepackage{lipsum}
\usepackage{titlesec}
\usepackage{dsfont}

\usepackage{scrextend}
\changefontsizes{10pt}

\usepackage[titletoc]{appendix}

\newcommand\id{\mathds{1}}
\DeclareMathOperator{\Tr}{Tr}

\theoremstyle{plain}
\newtheorem{Thm}{Theorem}[section]
\newtheorem{Pro}[Thm]{Proposition}

\theoremstyle{remark}

\begin{document}

\title{Information-theoretic constraints on correlations with indefinite causal order}

\author{Issam Ibnouhsein} \email{issam.ibnouhsein@gmail.com} \affiliation{CEA-Saclay/IRFU/LARSIM, 91191 Gif-sur-Yvette, France} \affiliation{Université Paris-Sud, 91405 Orsay, France}\affiliation{Quantmetry, 55 Rue La Boétie, 75008 Paris, France}
\author{Alexei Grinbaum} \affiliation{CEA-Saclay/IRFU/LARSIM, 91191 Gif-sur-Yvette, France}

\begin{abstract}
Reconstructions of quantum theory usually implicitly assume that experimental events are ordered within a global causal structure. The process matrix framework accommodates quantum correlations that violate an inequality verified by all causally ordered correlations. Using a generalized probabilistic framework, we propose three principles constraining bipartite correlations to the quantum bound. Our approach highlights the role of a measure of dependence other than mutual information for an information-theoretic reconstruction of causal structures in quantum theory.
\end{abstract}


\maketitle

\section{INTRODUCTION}

A physical theory is a coherent set of mathematical rules that correlate data recorded in experiments. Quantum theory is one such set of rules, however its different interpretations have produced no consensus on what these rules say about ``reality''. A different approach to understanding quantum theory is to modify some of its rules and compare predictions of the modified theory with those of the original. Previous attempts include quaternionic models \cite{adler} and a model with nonlinear terms in the Schrödinger equation \cite{weinberg2}. More recently, quantum information theory has triggered a new development: instead of modifying the set of mathematical rules of quantum theory, one tries to derive (a subset of) these rules from clear informational principles. Reconstructing quantum theory then means that one should look for clearly motivated constraints on the correlations between experimental records, such that they (partially) reproduce the predictions of the quantum formalism \cite{grinbjps}. For instance, general nonsignaling models such as the Popescu-Rohrlich boxes \cite{popescu_quantum_1994} have numerous properties in common with quantum theory, including no cloning \cite{masanes_general_2006,barrett_information_2007}, no broadcasting \cite{barnum_generalized_2007}, monogamy of correlations \cite{masanes_general_2006}, and information-disturbance trade-offs \cite{barrett_no_2005,scarani_secrecy_2006}. Nonetheless, some supraquantum models have powerful communication \cite{van_dam_nonlocality_2000,buhrman_nonlocality_2010} or nonlocal computation properties \cite{linden_quantum_2007} unobserved in nature. The set of quantum correlations is then partially derived from various principles such as relaxed uncertainty relations \cite{steeg_relaxed_2008,oppenheim_uncertainty_2010}, nonlocality swapping \cite{skrzypczyk_couplers_2009,skrzypczyk_emergence_2009}, macroscopic locality \cite{navascues_glance_2010}, and information causality  \cite{pawlowski_information_2009,allcock_recovering_2009}.

Various reconstructions of quantum theory \cite{zeilinger_foundational_1999,hardy_quantum_2001,spekkens_defense_2004,dakic_quantum_2009,chiribella_informational_2010,Masanes} assume, most often implicitly, that experimental events are ordered within a global causal structure. For example, Hardy proposed in \cite{hardy_quantum_2001} a reconstruction using as primitives the preparation, the transformation, and the measurement (PTM). Physical systems are defined in his reconstruction by two numbers: the number of degrees of freedom $K$, representing the minimum number of measurements to determine the state of the system; and the dimension $N$, corresponding to the maximum number of states perfectly distinguishable in one measurement of the system. The assumption of a global causal structure is encoded in how systems compose. Indeed, consider a composite system with subsystems $A$ and $B$. Hardy's fourth axiom expresses the operationally defined parameters $K_{AB}$ and $N_{AB}$ of the composite system in terms of the parameters of subsystems $A$ and $B$:
\begin{equation}
\nonumber
N_{AB}=N_A N_B, \hspace{0.3cm} K_{AB}=K_A K_B.
\end{equation}
This definition implies that only a superobserver can calculate $K_{AB}$ and $N_{AB}$, for it requires PTM on each subsystem even if $A$ and $B$ are not localized in the same laboratory. This in turn implies the existence of a global structure ordering PTM events, a problem already discussed by Hardy that lead him to build one of the first operational frameworks with no assumption of the existence of a global causal structure ordering events \cite{hardy_probability_2005,hardy_towards_2007}. To cite another example, Rovelli argued informally that quantumness follows from a limit on the amount of ``relevant'' information that can be extracted from a system~\cite{rovelli_relational_1996}. If the notion of relevance is to be connected to lattice orthomodularity in the quantum logical framework~\cite{grinbijqi}, the ensuing reconstruction of quantum theory will fundamentally depend on the order of binary questions asked to the system. For many systems, it requires the existence of a global causal structure ordering all incoming information.

Efforts in the direction initiated by Hardy were continued by Chiribella \emph{et al.} \cite{chiribella_quantum_2013} and Oreshkov \emph{et al.} \cite{oreshkov_quantum_2012}. We begin by presenting the latter framework in Sec.~\ref{sectionPMF}. A generalized notion of the quantum state, called ``process matrix", describes all possible correlations between two physical systems under the assumption that quantum theory is valid in local laboratories, but without assuming that these laboratories are embedded in a global causal structure. Certain correlations allowed by this framework violate a ``causal inequality" verified by all correlations between causally ordered events. The value of the bound on such correlations, which we call the ``quantum bound", was shown to be maximal for qubits and under a restricted set of local operations involving traceless binary observables \cite{brukner_bounding_2014}. In Sec.~\ref{sectionGPF}, we build a generalized probabilistic framework using primitives whose importance and relationships are discussed in detail. In Sec. \ref{sectionDPI} we show that the quantum bound can be derived within this probabilistic framework from a constraint on mutual information between parties that extends the usual data processing inequality (DPI) in a certain way. In Sec. \ref{sectionRAC}, we reformulate the Oreshkov \emph{et al.} causal game as a random access code (RAC) and define another class of causal games such that any protocol defined within a global causal structure obeys a tight information-theoretic inequality. Relaxing the signaling possibilities to the set of correlations obeying this information-theoretic inequality excludes supraquantum correlations and leads to a derivation of the quantum bound on correlations with indefinite causal order. Finally, we discuss alternative informational principles based on a measure of dependence other than mutual information which are able to distinguish between supraquantum, causally ordered, and quantum correlations with indefinite causal order. These results further contribute to our understanding of the causal structure of quantum theory via information-theoretic principles.

\section{THE PROCESS MATRIX FRAMEWORK}
\label{sectionPMF}

Consider a fixed number of laboratories equipped with random bit generators and observers capable of free choice. In each run of the experiment, each laboratory receives exactly one physical system, performs transformations allowed by quantum theory, and subsequently sends the system out. Suppose each laboratory is isolated from the rest of the world, except when it receives or emits the system.

\subsection{General framework} 

Denote the input and the output Hilbert spaces of Alice by $\mathcal{H}^{A_1}$ and $\mathcal{H}^{A_2}$ and those of Bob by $\mathcal{H}^{B_1}$ and $\mathcal{H}^{B_2}$. The sets of all possible outcomes of a quantum instrument at Alice's (Bob's) laboratory corresponds to the set of completely positive (CP) maps  $\{\mathcal{M}^{A_1 A_2}_i\}_{i=1}^n$ ($\{\mathcal{M}^{B_1 B_2}_j\}_{j=1}^n$). Using the Choi-Jamio\l kowsky isomorphism, we can express a CP map, $\mathcal{M}^{A_1 A_2}_i: \mathcal{L}(\mathcal{H}^{A_1})\longrightarrow\mathcal{L}(\mathcal{H}^{A_2})$, at Alice's laboratory as a positive semidefinite operator, $M^{A_1 A_2}_i$, acting on $\mathcal{H}^{A_1}\otimes \mathcal{H}^{A_2}$, and a CP map, $\mathcal{M}^{B_1 B_2}_j:\mathcal{L}(\mathcal{H}^{B_1})\longrightarrow\mathcal{L}(\mathcal{H}^{B_2})$, at Bob's laboratory as a positive semidefinite operator, $M^{B_1 B_2}_j$, acting on $\mathcal{H}^{B_1}\otimes \mathcal{H}^{B_2}$. Using this correspondence, the noncontextual probability for two measurement outcomes can be expressed as a bilinear function of the corresponding Choi-Jamio\l kowsky operators,
\begin{small}
\begin{equation}
\nonumber
P(\mathcal{M}^{A_1 A_2}_i,\mathcal{M}^{B_1 B_2}_j)=\Tr\left[W^{A_1 A_2 B_1 B_2} \left(M^{A_1 A_2}_i\otimes M^{B_1 B_2}_j\right) \right],
\end{equation}
\end{small}where $W^{A_1 A_2 B_1 B_2} \in \mathcal{L}(\mathcal{H}^{A_1}\otimes\mathcal{H}^{A_2}\otimes
\mathcal{H}^{B_1}\otimes\mathcal{H}^{B_2})$ is fixed for all runs of the experiment. Requiring that such probabilities be non-negative for any choice of CP maps, and equal to $1$ for any choice of CP and trace-preserving (CPTP) maps, yields a space of valid $W$ operators referred to as ``process matrices".

\subsection{The causal game} 

In this framework, two parties, Alice and Bob, each receive a system in their laboratory. Each of them tosses a coin, whose value is denoted $a$ for Alice and $b$ for Bob. They additionally share a random task bit $b'$ with the following meaning: if $b' = 0$, Bob must communicate $b$ to Alice; and if $b' = 1$, Bob must guess the value of $a$. Both parties always produce a guess, denoted $x$ for Alice and $y$ for Bob. It is crucial to assume that the bits $a$, $b$, and $b'$ are random.

The goal of Alice and Bob is to maximize the probability of success,
\begin{equation}
\nonumber
P_{success}=\frac{1}{2} \left[p(x=b|b'=0)+p(y=a|b'=1)\right],
\end{equation}
i.e. Alice should guess Bob's toss, or vice versa, depending on the value of $b'$. If all events occur in a causal sequence, then
\begin{equation}
\label{causineq}
P_{success}\leq \frac{3}{4}.
\end{equation}
Indeed, it is true that either Alice cannot signal to Bob or Bob cannot signal to Alice. Consider the latter case. If $b' = 1$, Alice and Bob could in principle achieve up to $p(y = a|b' = 1) = 1$. However, if $b' = 0$, Alice can only make a random guess, hence $p(x = b|b' = 0) = \frac{1}{2}$ and the probability of success in this case satisfies \eqref{causineq}. The same argument shows that the probability of success will not increase when Alice cannot signal to Bob or under any mixing strategy.

\bigskip

Now consider the following process matrix using the usual Pauli matrices $\sigma_x, \sigma_y$ and $\sigma_z$,
\begin{small}
\begin{equation}
\label{W-violation}
W^{A_1 A_2 B_1 B_2}=\frac{1}{4}\left[\id^{A_1 A_2 B_1 B_2}+\frac{1}{\sqrt{2}}\left(\sigma^{A_2}_z\sigma^{B_1}_z+\sigma^{A_1}_z\sigma^{B_1}_x\sigma^{B_2}_z\right)\right],
\end{equation}
\end{small}where $A_1,A_2,B_1$, and $B_2$ are two-level systems. Consider the following CP maps at Alice's and Bob's laboratories, respectively,
\begin{small}
\begin{equation}
\label{local-maps}
\begin{aligned}
\xi^{A_1 A_2}(x,a,b') & =\frac{1}{2}\left[\id+(-1)^x\sigma_z\right]^{A_1}\otimes \left[\id+(-1)^a\sigma_z\right]^{A_2},\\
\eta^{B_1 B_2}(y,b,b')&=b'\cdot \eta^{B_1 B_2}_1(y,b,b')  +(b'\oplus 1)\cdot \eta^{B_1 B_2}_2(y,b,b'),
\end{aligned}
\end{equation}
\end{small}where $\eta^{B_1 B_2}_1(y,b,b') =\frac{1}{2}\left[\id+(-1)^y\sigma_z\right]^{B_1}\otimes \id^{B_2}$ and $\eta^{B_1 B_2}_2(y,b,b')=\frac{1}{2}\left[\id^{B_1 B_2}+(-1)^b\sigma_x^{B_1}\sigma_z^{B_2}\right]$. Computations show that the success probability associated with \eqref{W-violation} and \eqref{local-maps} violates causal inequality \eqref{causineq}:
\begin{equation}
\label{quantum-bound}
P_{success}=\frac{2+\sqrt{2}}{4}>\frac{3}{4}.
\end{equation}
Hence it is impossible to interpret these events as occurring within a global causal structure. This is an example of a causally nonseparable process, \emph{viz.}, a process that cannot be written as (a mixture of) causal processes,
\begin{equation}
W\neq \lambda W^{A \npreceq  B}+(1-\lambda )W^{B\npreceq  A},\label{causform}
\end{equation}where $0\leq \lambda \leq 1$, $W^{A\npreceq  B}$ is a process in which Alice cannot signal to Bob and $W^{B\npreceq  A}$ a process in which Bob cannot signal to Alice. ``Cannot signal'' here means either that the channels go in the other direction or that parties share a bipartite state. If a process matrix $W$ can be written in the form \eqref{causform}, it will be called ``causally separable".

\section{GENERALIZED PROBABILISTIC FRAMEWORK}
\label{sectionGPF}

We now aim at building a probabilistic framework using only the input bits, the output bits, a random task bit, and the notion of causal order as primitives to describe the experiment that violates the causal inequality. Using bits $x,y,a$ and $b$ defined in the previous section is necessary for the construction of this framework: they correspond to the information a party wants to send or to the result of a measurement on the received system that allows a party to retrieve the information the distant party transmitted. Similarly, the notion of causal order is necessary to establish the inequalities one is trying to violate using quantum theory. On the contrary, the role of the random task bit $b'$ and whether this bit should be included as a primitive in the probabilistic framework needs to be clarified.

Consider, as in the previous game, two parties Alice and Bob with inputs $a,b$ and outputs $x,y$ with obvious notations. Bob also possesses a random task bit, $b'$. Now, suppose we are given a quantum process matrix and a strategy (with local quantum operations) by means of which we realize a specific joint probability distribution $p(x,y|a,b)$ after tracing over the random task bit $b'$,
\begin{equation}
p(x,y|a,b) = \sum_{\alpha} p(x,y|a,b,b'=\alpha)p(b'=\alpha),
\end{equation}thus yielding a new effective strategy. We show that if for each fixed value $\alpha$ of $b'$, $p(x,y|a,b,b'=\alpha)$ can be realized using \emph{fixed} local quantum instruments, i.e. independent of $a$ and $b$, then there exists an equivalent diagonal quantum process by means of which we obtain the same probabilities $p(x,y|a,b)$ for all $a,b,x,y$. Since a diagonal bipartite process is causally separable \cite{oreshkov_quantum_2012}, $p(x,y|a,b)$ arising from such an effective strategy cannot violate \emph{any} causal inequality. It is crucial for the argument that the effective local operations can be taken to be diagonal in a \emph{fixed} local basis so that there exists a \emph{single} diagonal process matrix that yields the joint probabilities \emph{for all $a,b,x, y$}. Obviously, if $x$ and $y$ are produced before $a$ and $b$ and after $b'$, then the quantum instruments whose outcomes yield $x$ and $y$ cannot depend on $a$ and $b$, and hence can be considered as fixed for each fixed value of $b'$. Since the operations \eqref{local-maps}, allowing a violation of the causal inequality verify the following local ordering constraint
\begin{equation}
\label{local-ordering-constraint}
 b'  \preceq    y \preceq b \quad \mbox{ and } \quad x  \preceq a,
 \end{equation}
they also verify the constraint of fixed local quantum instruments for each fixed value of $b'$ and therefore cannot violate any causal inequality using only bits $x,y,a$ and $b$. Hence the necessity to include the random task bit $b'$ as a primitive in the generalized framework if we are to describe violations of the causal inequality under the local ordering constraint \eqref{local-ordering-constraint}.

We now prove our initial assumption. Consider a fixed value $\alpha$ of $b'$:
\begin{itemize}
\item[(a)] By assumption, the most general strategy for Bob is to apply a fixed quantum instrument denoted $I_1(\alpha)$ to the input system, whose outcome yields $y$, and to subject the output system of that instrument to a subsequent CPTP map dependent on the value of $b$, denoted $I_2(\alpha,b)$. 
\item[(b)] The first quantum instrument $I_1(\alpha)$ can be implemented by a unitary $U_1(\alpha)$ on the input system plus an ancilla, followed by a projective measurement $P(\alpha)$ on part of the resulting joint system \cite{oreshkov_quantum_2012}. The CPTP map $I_2(\alpha,b)$ can be implemented by a unitary $U_2(\alpha)$ applied to the output of $I_1(\alpha)$, an ancilla, and a qubit prepared in the state $|b\rangle$ (we feed $b$ in the form of a quantum state $|b\rangle$, where different vectors $|b\rangle$ are orthogonal). 
\item[(c)] The projective measurement $P(\alpha)$ and the preparation of $|b\rangle$ fully define Bob's operation: other transformations as well as the ancillae can be seen as occuring outside Bob's laboratory by attaching them to the original process \emph{before the input}, which yields a new equivalent process with a new process matrix that depends on $\alpha$ (note that here lies the aforementioned connection between an effective \emph{fixed} strategy for each value of $b'$ and the existence of a \emph{single} effective process: if the first local unitary before the projective measurement depends on $a$ or $b$, then for each particular value of $a$ or $b$ we can pull it out of the laboratory before the input system, but this does not yield one \emph{single} quantum process from which $p(x,y|a,b)$ is obtained with diagonal operations \emph{for all} $a,b,x$, and $y$). If the original matrix were valid, then whatever Bob may choose to do on his redefined input and output systems could have occurred anyway and would have yielded valid probabilities, hence the redefined process matrix is also valid. Here we focused on operations in Bob's laboratory, but similar arguments hold for operations in Alice's laboratory (which are independent from $b'$).  As a result, we obtain that the correlations for each value $\alpha$ of $b'$, and hence for the effective (mixed) strategy, are equivalent to correlations obtained by diagonal measurement and repreparation operations, i.e. classical local operations. 
\end{itemize}

In the remainder of the paper, when referring to runs of the Oreshkov \emph{et al.} game we are considering Alice and Bob in the context of the generalized probabilistic framework that uses as primitives the input bits, the output bits, a random task bit, and the notion of causal order. 

\section{MULTICHANNEL DATA PROCESSING INEQUALITY} 
\label{sectionDPI}

In this section, we derive the quantum bound on correlations with indefinite causal order from a constraint on mutual information Alice and Bob share very similar to a DPI. If the causal order between Alice and Bob is well defined, either Alice cannot signal to Bob or Bob cannot signal to Alice, a constraint that can be formulated using a measure of dependence (in the sense of Rényi \cite{renyi_new_1959,renyi_measures_1959}) such as mutual information as follows:
\begin{equation}
\label{one-con}
I(x:b|b'=0)+I(y:a|b'=1)\leq 1.
\end{equation}However, this condition is not sufficient for limiting correlations to the ones allowed by the process matrix framework because
\begin{small}
\begin{align}
I(x:b|b'=0)&=\frac{1+E_1}{2}\log_2(1+E_1)+\frac{1-E_1}{2}\log_2(1-E_1),\\
I(y:a|b'=1)&=\frac{1+E_2}{2}\log_2(1+E_2)+\frac{1-E_2}{2}\log_2(1-E_2),
\end{align}
\end{small}where $p(b\oplus x = 0|b'=0) =\frac{1+E_1}{2}$ and $p(a\oplus y= 0|b'=1) = \frac{1+E_2}{2}$, and one can show that there are supraquantum correlations with $E_1^2+E_2^2>1$ that verify \eqref{one-con}. Consequently, stronger constraints are needed.

\begin{Pro}
\label{prop-box}
Consider two independent runs of the Oreshkov \emph{et al.} game $(E_1^{(1)},E_2^{(1)},x_1,y_1,a_1,b_1,b'_1)$ and $(E_1^{(2)},E_2^{(2)},x_2,y_2,a_2,b_2,b'_2)$, where:
\begin{align}
\begin{aligned}
\label{assumption}
p(b_i\oplus x_i= 0|b'_i=0) & =\frac{1+E_1^{(i)}}{2},\\
p(a_i \oplus y_i= 0|b'_i=1) & = \frac{1+E_2^{(i)}}{2},\quad i=1,2.
\end{aligned}
\end{align}
The following two conditions are equivalent.

Condition (i):
\begin{equation}
\label{hypo1}
\begin{aligned}
I(x_1:b_1|b'_1=0) &\geq I(x_1\oplus x_2:b_1\oplus b_2|b'_1=0,b'_2=0)\\
&+I(x_1\oplus y_2:b_1\oplus a_2|b'_1=0, b'_2=1),
\end{aligned}
\end{equation}
\begin{equation}
\label{hypo2}
\begin{aligned}
I(y_1:a_1|b'_1=1) &\geq I(y_1\oplus x_2:a_1\oplus b_2|b'_1=1,b'_2=0)\\
&+I(y_1\oplus y_2:a_1\oplus a_2|b'_1=1, b'_2=1).
\end{aligned}
\end{equation}

Condition (ii):
\begin{equation}\label{hypo3}
(E_1^{(2)})^2+(E_2^{(2)})^2 \leq 1. 
\end{equation}
\end{Pro}
\begin{proof}
Suppose that $(E_1^{(2)})^2+(E_2^{(2)})^2 \leq 1$ holds. Define the variables
\begin{align}
\begin{aligned}
 X& =b_1|[b'_1=0], \\
 Y&=x_1|[b'_1=0],\\
  Z&=x_1\oplus x_2 \oplus b_2|[b'_1=0,b'_2=0],
 \end{aligned}
 \end{align} where the entire expression on the left-hand side of the bar `$|$' is conditioned by the expression in brackets on the right-hand side. Because the two runs are assumed to be independent, one can see that 
\begin{equation}
X\rightarrow Y \rightarrow Z
\end{equation}
 is a Markov chain with transition parameters $p_1=\frac{1+E_1^{(1)}}{2}$ and $p_2=\frac{1+E_1^{(2)}}{2}$, therefore a strong form of the DPI applies \cite{Erkip98theefficiency},
\begin{equation}
\label{mrs}
I(X:Z)\leq \rho^*(Y:Z)^2 I(X:Y),
\end{equation}
where $\rho^*(Y:Z)$ defines the Hirschfeld-Gebelein-Rényi maximal correlation of variables $Y$ and $Z$ \cite{hirschfeld_connection_1935,Gebelein,renyi_new_1959,renyi_measures_1959}. Since $Y,Z$ are Bernoulli variables, we have $\rho^*(Y:Z)=2p_2-1=E_1^{(2)}$; therefore,
\begin{equation}
 I(x_1\oplus x_2:b_1\oplus b_2|b'_1=0,b'_2=0) \leq (E_1^{(2)})^2 I(x_1:b_1|b'_1=0).
 \end{equation}
Similarly, one can show that
\begin{equation}
 I(x_1\oplus y_2:b_1\oplus a_2|b'_1=0, b'_2=1)\leq (E_2^{(2)})^2  I(x_1:b_1|b'_1=0).
 \end{equation}
Therefore imposing $(E_1^{(2)})^2+(E_2^{(2)})^2 \leq 1$  implies \eqref{hypo1}. One can similarly show that $(E_1^{(2)})^2+(E_2^{(2)})^2 \leq 1$ also implies \eqref{hypo2}.

To prove the converse, we recall that since $Y,Z$ are Bernoulli variables, we have \cite{anantharam_maximal_2013}
\begin{equation}
\rho^*(Y:Z)^2=\underset{X\rightarrow Y \rightarrow Z}{\sup} \frac{I(X:Z)}{I(X:Y)}.
\end{equation}Using \eqref{mrs}, one can show that \eqref{hypo1} and \eqref{hypo2} imply \eqref{hypo3}.
\end{proof}

If and only if a causal order is fixed, Eqs. \eqref{hypo1} and \eqref{hypo2} take the form of the usual DPI. In general, however, these equations involve sums of variables from two possible causal orders for a single round, while the DPI requires that information be discarded in a fixed direction. Consequently, this alternative approach leads to two original conditions but their significance is blurred by the intertwining of causal orders.

\section{CAUSAL GAMES AS RANDOM ACCESS CODES}
\label{sectionRAC}

In this section, we reformulate the causal game as a distributed RAC. This is motivated by the RAC formulation of the game for which the information causality principle was introduced \cite{pawlowski_information_2009}. Such an approach might open the path for a formulation of an analog of the information causality principle in the context of causal games. 

\subsection{Reformulation of the causal game} 

Consider two independent runs of the experiment described in the Oreshkov \emph{et al.} game, with bits $\{x_1,y_1,a_1,b_1\}$ and $\{x_2,y_2,a_2,b_2\}$, respectively. The random task bit $b'$ now corresponds to a pair of bits $b'_1b'_2$ denoting the four possible combinations of tasks for two runs of the experiment: $b'=0_1 0_2$ means that in both runs Alice must guess Bob's bit, $b'=0_1 1_2$ means that Alice must guess Bob's bit in the first run and Bob must guess Alice's bit in the second run, and so forth. It is straightforward to generalize this notation for $n$ runs.

Assume that different runs of the experiment use the same box as a resource:
\begin{align}
\begin{aligned}
p(b_i\oplus x_i= 0|b'_i=0) &= p(b_j\oplus x_j= 0|b'_j=0), \\
p(a_i\oplus y_i= 0|b'_i=1) &= p(a_j\oplus y_j= 0|b'_j=1),  \quad \forall i,j.
\end{aligned}
\end{align}Again, we write
\begin{align}
\begin{aligned}
\label{assumption}
p(b_i\oplus x_i= 0|b'_i=0) &=\frac{1+E_1}{2},\\
p(a_i \oplus y_i= 0|b'_i=1) &= \frac{1+E_2}{2},\quad \forall i.
\end{aligned}
\end{align}Now consider $n$ runs of the experiment and define
\begin{equation}
\label{optimize-n}
\begin{aligned}
P_n =& \frac{1}{2^n}\left[p(b_1\oplus x_1 \oplus..\oplus b_n \oplus x_n=0|b'=0_10_2..0_{n})\right.\\
&+\left. p(b_1\oplus x_1 \oplus..\oplus b_{n-1} \oplus x_{n-1} \oplus a_n \right.\\
&\left. \oplus y_n=0| b'=0_1..0_{n-1}1_{n})\right.\\
&\left. +...+p(a_1\oplus y_1 \oplus..\oplus a_n \right.\\
&\left. \oplus x_n=0|b'=1_11_2..1_{n})\right].\\
\end{aligned}
\end{equation}For each term in brackets, the condition that the sum over the guesses for $n$ runs means that either both Alice and Bob make an even number of mistakes or both make an odd number of mistakes. We now compute the expression of a term $p_{n-k,k}$ inside the brackets for which the number of 0's in $b'$ is $n-k$ and the number of 1's is $k$. 

The probability of an even number of wrong guesses by Alice is
\begin{small}
\begin{align}
\begin{aligned}
Q^{(n-k)}_{even}(\mbox{Alice}) &= \sum\limits_{j=1}^{\lfloor \frac{n-k}{2}\rfloor}\binom{n-k}{2j} \left(\frac{1-E_1}{2}\right)^{2j}\left(\frac{1+E_1}{2}\right)^{n-k-2j}\\
&=\frac{1+E_1^{n-k}}{2}.
\end{aligned}
\end{align}
\end{small}Similarly, the probability of an odd number of wrong guesses by Alice is
\begin{small}
\begin{align}
\begin{aligned}
Q^{(n-k)}_{odd}(\mbox{Alice})&= \sum\limits_{j=1}^{\lfloor\frac{n- k-1}{2}\rfloor}\binom{n-k}{2j+1} \left(\frac{1-E_1}{2}\right)^{2j+1}\\
&\cdot \left(\frac{1+E_1}{2}\right)^{n-k-2j-1} =\frac{1-E_1^{n-k}}{2}.
\end{aligned}
\end{align}
\end{small}The probability of an even number of wrong guesses by Bob is
\begin{small}
\begin{align}
\begin{aligned}
Q^{(k)}_{even}(\mbox{Bob}) &= \sum\limits_{j=1}^{\lfloor \frac{k}{2}\rfloor}\binom{k}{2j} \left(\frac{1-E_2}{2}\right)^{2j}\left(\frac{1+E_2}{2}\right)^{k-2j}\\
&=\frac{1+E_2^k}{2}.
\end{aligned}
\end{align}
\end{small}Similarly, the probability of an odd number of wrong guesses by Bob is
\begin{small}
\begin{align}
\begin{aligned}
Q^{(k)}_{odd}(\mbox{Bob})&= \sum\limits_{j=1}^{\lfloor\frac{ k-1}{2}\rfloor}\binom{k}{2j+1} \left(\frac{1-E_2}{2}\right)^{2j+1}\left(\frac{1+E_2}{2}\right)^{k-2j-1}\\
&=\frac{1-E_2^k}{2}.
\end{aligned}
\end{align}
\end{small}The final expression for a term inside the brackets where the number of 1's in $b'$ is $k$ is
\begin{small}
\begin{align}
\begin{aligned}
& p_{n-k,k} \\
&=Q^{(n-k)}_{even}(\mbox{Alice})\cdot Q^{(k)}_{even}(\mbox{Bob})+Q^{(n-k)}_{odd}(\mbox{Alice})\cdot Q^{(k)}_{odd}(\mbox{Bob})\\
&= \frac{1}{2}\left[1+E_1^{n-k}E_2^k\right],
\end{aligned}
\end{align}
\end{small}and 
\begin{equation}
\label{rac_prob}
P_n=\frac{1}{2^n}\sum_{k=0}^{2^n-1}\binom{n}{k} p_{n-k,k}.
\end{equation}We now treat the two bits in $b'$ as binary notation of an integer and identify $b^\prime$ with this integer. For example, when $n=2$, $b'=01$ corresponds to 1 and $b^\prime = 10$ to 2. For a given decimal $b^\prime=i$, we group the runs by specifying an expression to be set to 0, which we denote $g_i \oplus t_i = 0$, where $g_i$ is the sum of output bits (`guesses') and $t_i$ the sum of input bits (`tosses'). To continue the example $n=2$, for $b'=1$ we set $x_1\oplus b_1 \oplus y_2\oplus a_2=0$ with the bit of guesses $g_1=x_1\oplus y_2$ and the bit of tosses $t_1=b_1\oplus a_2$. For $b'=2$ the corresponding expression is $y_1\oplus a_1 \oplus x_2\oplus b_2=0$ with the bit of guesses $g_2=y_1\oplus x_2$ and the bit of tosses $t_2=a_1\oplus b_2$. This provides a reformulation of the causal game as a RAC, where Eq. \eqref{rac_prob} is the probability one wants to maximize.

\subsection{Information-theoretic inequality for causal structures}

Using the same notation as in the previous section, we now introduce an information-theoretic inequality verified by all events occuring within a global causal structure. First, we need to demonstrate the following result.
\begin{Pro}
The following inequality holds:
\begin{equation}
\label{cond1}
\frac{(E_1^2+E_2^2)^{n}}{2\ln(2)} \leq I(n) \leq (E_1^2+E_2^2)^n,
\end{equation}where $I(n)=\sum_{i=0}^{2^n-1} I(g_i:t_i|b'=i)$ is a measure of efficiency of the $n$ runs protocol, $I(X:Y)$ denotes mutual information between random variables $X$ and $Y$, and $h$ is the binary entropy.
\end{Pro}
\begin{proof}
Recalling that $n-k$ is the number of 0's and $k$ is the number of 1's in the n bits binary notation of integer $i$, we have
\begin{equation}
\label{form}
\begin{aligned}
I(g_i:t_i|b'=i)&=\frac{1+E_1^{n-k}E_2^k}{2}\log_2 (1+E_1^{n-k}E_2^k) \\
&+ \frac{1-E_1^{n-k}E_2^k}{2}\log_2 (1-E_1^{n-k}E_2^k)\\
& = 1 - h\left(\frac{1}{2}(1+E_1^{n-k}E_2^k)\right)\\
& \geq \frac{(E_1^2)^{n-k}(E_2^2)^k}{2\ln 2},
\end{aligned}
\end{equation}where we used $h(\frac{1}{2}(1+y))\leq 1-\frac{y^2}{2 \ln 2}$. Therefore
\begin{align}
\begin{aligned}
I(n) & = \sum_{i=0}^{2^n-1}I(g_i:t_i|b'=i) \\
&\geq\frac{1}{2\ln 2 }\sum_{k=0}^{n} \binom{n}{k}(E_1^2)^{n-k}(E_2^2)^k\\
&=\frac{1}{2\ln 2}(E_1^2+E_2^2)^n.
\end{aligned}
\end{align}
We also have
\begin{equation}
\label{form}
\begin{aligned}
I(g_i:t_i|b'=i)&=\frac{1+E_1^{n-k}E_2^k}{2}\log_2 (1+E_1^{n-k}E_2^k) \\
&+ \frac{1-E_1^{n-k}E_2^k}{2}\log_2 (1-E_1^{n-k}E_2^k)\\
&\leq (E_1^2)^{n-k}(E_2^2)^k,
\end{aligned}
\end{equation}therefore
\begin{equation}
I(n) \leq \sum_{k=0}^n \binom{n}{k} (E_1^2)^{n-k}(E_2^2)^k = (E_1^2+E_2^2)^n.
\end{equation}
\end{proof}

We claim that any causally separable process verifies
\begin{equation}
\label{caus-ineq}
I(n)\leq 1, \quad \forall n,
\end{equation}and that 1 is the only nonzero bound on $I(n)$. To see this, consider a fixed causal structure and a given value $b' = i$. Then all $g_j \oplus t_j, j \neq i$, are equal to 0 with probability $\frac{1}{2}$; therefore $I(g_i:t_i|b'=i)\leq 1$ and $I(g_j:t_j|b'=j)=0$ for $j\neq i$, leading to $I(n)\leq 1$. The mutual information expression $I(X:Y|Z)$ is convex in $p(y|x,z)$, where $x,y$ and $z$ are values that the random variables $X,Y$ and $Z$ can, respectively, take, and therefore no mixture of strategies with fixed causal structures can increase the value of $I(n)$. Consequently, inequality \eqref{caus-ineq} is valid for all causally separable processes, and since condition $I(n) = 1, \forall n$, can be reached using a fixed causal structure, it is a \emph{tight} inequality.

We now ask whether this condition fully characterizes the set of causally separable processes and, if not, what the exact set of correlations verifying such a condition is. If the bounded efficiency condition is taken as a constraint on the correlations between Alice's and Bob's laboratories, one can show that the set of allowed correlations is those respecting the quantum bound. Indeed, Eq. \eqref{cond1} shows that a limit on the protocol efficiency \emph{for any number of runs} is equivalent to the bound 1 on $E_1^2+E_2^2$ or, equivalently, to the bound $\frac{1}{\sqrt{2}}$ on $E$, where $E=E_1=E_2$ if all probabilities \eqref{assumption}, are equal. Therefore, relaxing the signaling possibilities to the set of correlations obeying the tight information-theoretic inequality \eqref{caus-ineq}, verified by all events occuring within a global causal structure, allows us to retrieve the quantum bound on correlations with indefinite causal order.

This result is somewhat analogous to the principle of information causality where, given a set of ``classical" resources (shared nonsignaling correlations and one-way signaling) and a class of games (increasing size of Alice's data set), one can derive the quantum bound on correlations by keeping the same information-theoretic figure of merit quantifiying the performance of the parties in winning such games for classical and quantum resources. Note that this similarity is only intuitive and not at all rigorous, because in the context of no-signaling games, one can show that the principle of information causality is distinct from the ``no-supersignaling" principle, which encodes the idea that the protocol efficiency must not increase \cite{wakakuwa_chain_2012}. The main obstacle to a direct transposition of the proof of \cite{pawlowski_information_2009} to the RAC formulation of the causal game is the dependence between the guesses expressions $g_i$ (and similarly for tosses expressions $t_i$).

\subsection{Beyond mutual information}

Shifting the focus from mutual information to another measure of dependence, one can easily check that conditions \eqref{hypo1} and \eqref{hypo2} (or, alternatively, the quantum bound) are equivalent to imposing
\begin{equation}
\label{causal-ineq}
\rho^*(Y:Z)^2+\rho^*(Y:Z')^2\leq 1,
\end{equation}where we kept the notation from the corresponding proof and defined
\begin{equation}
 Z'=x_1\oplus y_2\oplus a_2|[b'_1=0,b'_2=1].
\end{equation}
 More generally, the quantum bound is equivalent to the following constraint
\begin{equation}
\label{causal-ineq2}
\rho^*(x|b'=0:b|b'=0)^2+\rho^*(y|b'=1:a|b'=1)^2\leq 1,
\end{equation}while causally separable processes are characterized by
\begin{equation}
\label{causal-ineq3}
\rho^*(x|b'=0:b|b'=0)+\rho^*(y|b'=1:a|b'=1)\leq 1.
\end{equation}
Since the HGR maximal correlation is also a measure of dependence, Eq. \eqref{causal-ineq3} has the same clear informational interpretation in terms of allowed signaling directions between parties within (a mixture of) fixed causal structures as Eq. \eqref{one-con}. The square of the HGR maximal correlation of Bernoulli variables, which appears in \eqref{causal-ineq2}, also has an information-theoretic interpretation: it quantifies the initial efficiency of communication between parties \cite{Erkip98theefficiency}. Indeed, taking $Y=x|[b'=0]$ and $Z=b|[b'=0]$ we obtain
\begin{equation}
\label{connect-HGR-MI}
\rho^*(Y:Z)^2=\Delta^\prime(0),
\end{equation}where $\Delta^\prime$ is the derivative of 
\begin{equation}
\Delta(R)= \sup\limits_{\substack{X \rightarrow Y \rightarrow Z\\I(X:Y)\leq R}} I(X:Z).
\end{equation}
Thus condition \eqref{causal-ineq2} means that the dependence between parties can exceed one bit as long as the total initial efficiency of communication does not exceed one bit.

In summary, equality \eqref{connect-HGR-MI} connects the HGR maximal correlation and the increase in mutual information. It is based on inequality \eqref{mrs} and is central to an information-theoretic interpretation of condition \eqref{causal-ineq2}. To prove \eqref{cond1} or \eqref{mrs}, one uses the standard symmetry, non-negativity, chain rule and data processing properties of mutual information. Therefore, the bound on quantum correlations with indefinite causal order is equivalent to imposing these standard properties on mutual information between inputs and outputs of parties, with an additional consistency condition for classical systems, so that mutual information between independent systems equals 0, along with one of condition \eqref{caus-ineq} or \eqref{causal-ineq2}.

\section{CONCLUSION}

We defined a generalized probabilistic framework to discuss the connection between the quantum bound on correlations with indefinite causal order and various information-theoretic principles. We have shown that the quantum bound on the causal game can be derived from a constraint on the mutual information shared by Alice and Bob that extends the usual DPI in a certain way. We have reformulated the causal game as a RAC and defined a new class of causal games for which all causally separable processes obey a tight information-theoretic inequality. By relaxing the signaling possibilities to the set of correlations that obey this information-theoretic inequality we retrieve the quantum bound on correlations with indefinite causal order. Using an alternative measure of dependence, we establish a relationship between the quantum bound and the initial efficiency of communication. Central to these derivations are standard properties of mutual information. Our approach highlights both qualitatively and quantitatively the fact that mutual information may not be the most convenient measure of dependence for causal games. Whether ``natural'' properties of alternative measures, e.g., HGR maximal correlation, lead to the quantum bound under the local ordering condition and for more general operations is currently under investigation.

\section*{ACKNOWLEDGMENTS}
We thank O. Oreshkov, F. Costa, M. Paw\l{}owski and Ä. Baumeler for helpful discussions. This work was supported by the European Commission Project Q-ESSENCE (No.248095).

\bibliographystyle{apsrev4-1}
\bibliography{paper}

\end{document}